\documentclass[journal,twoside,web]{ieeecolor}
\usepackage{lcsys}
\usepackage[utf8]{inputenc}

\usepackage{balance}

\usepackage{multirow}

\usepackage{amsmath,mathtools}  %
\usepackage{amsfonts,amssymb}
\usepackage{bbm}

\usepackage{xspace}    %

\usepackage{graphicx}
\graphicspath{{figs/}}

\usepackage{subcaption}
\usepackage{caption}

\usepackage{hyperref}
\hypersetup{colorlinks=true, linkcolor= blue, allcolors=blue, citecolor = blue, filecolor=black, urlcolor=blue}

\usepackage{todonotes}
\usepackage{comment}

\usepackage{tikz}
\usetikzlibrary{arrows,positioning}   %

\usepackage{pgfplots}\pgfplotsset{compat=1.9}
\usepackage{grffile}
\pgfplotsset{compat=newest}  %
\usetikzlibrary{plotmarks}
\usetikzlibrary{arrows.meta}
\usepgfplotslibrary{patchplots}
\usepgfplotslibrary{groupplots}

\pgfplotsset{compat=1.15}
\usepackage{mathrsfs}
\usetikzlibrary{arrows}

\usepackage{tabularx}

\usepackage[ruled]{algorithm}
\usepackage[noend]{algpseudocode}

\usepackage{diagbox}
\newcolumntype{K}[1]{>{\centering\arraybackslash}p{#1}}

\newcommand{\hilitediff}[1]{#1}

\makeatletter
\newcommand{\multiline}[1]{%
  \begin{tabularx}{\dimexpr\linewidth-\ALG@thistlm}[t]{@{}X@{}}
    #1
  \end{tabularx}
}
\makeatother

\usepackage[per-mode=symbol]{siunitx}

\usepackage{cite}  %
\usepackage{array}  %
\usepackage{booktabs}           %

\newtheorem{theorem}{Theorem}

\newtheorem{lemma}[theorem]{Lemma}

\newtheorem{assumption}{Assumption}

\newcounter{remark}
\newenvironment{remark}{%
\par\vspace{3pt}\noindent\refstepcounter{remark}\textbf{Remark~\theremark:}}%
{\par\endtrivlist\unskip}

\newtheorem{definition}{Definition}

\newcounter{problem}
{\par\endtrivlist\unskip}

\usepackage[extramath,probability,reviewmark]{mylatexdefs}

\newcommand{\probP}{\text{I\kern-0.15em P}}

\newcommand{\CAV}[1]{CAV\textendash\ensuremath{(#1)}\xspace}

\newcommand{\veh}[1]{vehicle\textendash\ensuremath{(#1)}\xspace}
\newcommand{\lane}[1]{lane\textendash\ensuremath{#1}\xspace}

\newcommand{\TLC}[1]{TLC\textendash\ensuremath{#1}\xspace}
\newcommand{\node}[1]{node\textendash\ensuremath{#1}\xspace}
\newcommand{\agent}[1]{agent\textendash\ensuremath{#1}\xspace}
\newcommand{\Agent}[1]{Agent\textendash\ensuremath{#1}\xspace}

\newcommand{\bbsym}[1]{\ensuremath{\boldsymbol{#1}}}

\usepackage{stfloats}%

\def\BibTeX{{\rm B\kern-.05em{\sc i\kern-.025em b}\kern-.08em
    T\kern-.1667em\lower.7ex\hbox{E}\kern-.125emX}}
\begin{document}

\title{
Distributed Optimization for Traffic Light Control and Connected Automated Vehicle Coordination in Mixed-Traffic Intersections}
\author{Viet-Anh Le$^{1,2}$, {\IEEEmembership{Student Member, IEEE}}, and Andreas A. Malikopoulos$^{3}$, {\IEEEmembership{Senior Member, IEEE}}
\thanks{This research was supported in part by NSF under Grants CNS-2149520, CMMI-2219761, IIS-2415478, and in part by Mathworks.}
\thanks{$^{1}$Department of Mechanical Engineering, University of Delaware, Newark, DE 19716 USA.}
\thanks{$^{2}$Systems Engineering Field, Cornell University, Ithaca, NY 14850 USA.}
\thanks{$^{3}$School of Civil and Environmental Engineering, Cornell University, Ithaca, NY 14853 USA.}
\thanks{Emails: {\tt\small \{vl299,amaliko\}@cornell.edu}.}
}

\maketitle
\thispagestyle{empty}
\pagestyle{empty}

\begin{abstract}

In this paper, we consider the problem of coordinating traffic light systems and connected automated vehicles (CAVs) in mixed-traffic intersections.
We aim to develop an optimization-based control framework that leverages both the coordination capabilities of CAVs at higher penetration rates and intelligent traffic management using traffic lights at lower penetration rates.
Since the resulting optimization problem is a multi-agent mixed-integer quadratic program, we propose a penalization-enhanced maximum block improvement algorithm to solve the problem in a distributed manner.
The proposed algorithm, under certain mild conditions, yields a feasible person-by-person optimal solution of the centralized problem.
The performance of the control framework and the distributed algorithm is validated through simulations across various penetration rates and traffic volumes.

\end{abstract}

\begin{IEEEkeywords}
Connected automated vehicles, traffic light control, distributed mixed-integer optimization.
\end{IEEEkeywords}

\section{Introduction}

\IEEEPARstart{T}{hough} coordination of connected automated vehicles (CAVs) has shown promise in improving various traffic metrics \cite{Malikopoulos2020,chalaki2020TCST}, achieving full CAV penetration is not expected in the near future. 
Therefore, addressing the planning and control problem for CAVs in mixed traffic, where human-driven vehicles (HDVs) coexist, remains a fundamental research area.
Our early work on that topic focused on the problems in unsignalized mixed-traffic scenarios such as merging at roadways \cite{le2024stochastic} or single-lane intersections \cite{Le2023ACC}.
However, in more complex scenarios such as multi-lane intersections, coordinating CAVs alongside multiple HDVs without traffic signals may be unrealistic.
As intelligent traffic light systems have been commonly used to exert control over HDVs at intersections, the idea of combining traffic signal control and CAVs in mixed traffic has gained growing attention in recent years \cite{li2023survey}.
The current state-of-the-art methods for coordination of traffic light systems and CAVs in mixed-traffic intersections can be classified into three main categories: (1) \emph{reinforcement learning}, (2) \emph{bi-level optimization}, and (3) \emph{joint optimization}.
Reinforcement learning \cite{guo2023cotv,song2021traffic} aims to train control policies that optimize a specific reward function. 
However, real-time safety implications are not taken into account.
Optimization-based approaches, on the other hand, are well-studied for their ability to ensure real-time safety against the uncertain behavior of HDVs. 
Bi-level optimization approaches \cite{kamal2019development,du2021coupled} separate the traffic signal optimization from the CAV trajectory optimization. 
First, the traffic signal optimization is solved using an approximate traffic model, followed by solving the CAV trajectory optimization.
In joint optimization approaches \cite{ghosh2022traffic,tajalli2021traffic,ravikumar2021mixed,firoozi2022coordination,suriyarachchi2023optimization}, optimization problems considering both traffic signal control and CAV trajectory optimization are formulated.

However, the aforementioned research efforts considered traffic light control with CAV trajectory optimization rather than fully exploiting the coordination between CAVs for crossing the intersections. 
Specifically, in prior work, the traffic lights for lanes with lateral conflicts could not be green simultaneously, even if all the vehicles were CAVs. 
This constraint may hinder traffic efficiency at high CAV penetration rates, as it has been shown that coordinating CAVs can significantly improve throughput and energy efficiency while ensuring safety without the need for traffic signals, \eg \cite{Malikopoulos2020,chalaki2020TCST}.
Furthermore, in the existing optimization-based approaches, the optimization problems are solved in a centralized manner, which may not be scalable with an increasing number of CAVs.
Therefore, in this paper, we aim to make the following contributions.
First, we formulate a joint optimization problem in which the traffic lights for the conflicting lanes can be green simultaneously if there are no lateral conflicts between an HDV and a CAV or between two HDVs.
As a result, we can better leverage CAV coordination to enhance traffic throughput. %
The resulting optimization problem is a \emph{multi-agent mixed-integer quadratic program} (MIQP) that is challenging to solve in a distributed manner. 
To address this, we propose a variant of the maximum block improvement (MBI) algorithm \cite{chen2012maximum} where the local constraints of each agent are encoded in the objective using penalty functions.
\hilitediff{We show that the algorithm can find a \emph{person-by-person optimal solution} of the multi-agent MIQP problem (equivalent to Nash equilibrium in a game-theoretic context),} \ie no agent can unilaterally improve its local objective without changing the solution of at least one neighboring agent that shares coupling constraints.
We validate the proposed framework by simulations in a microscopic traffic simulator \cite{lopez2018microscopic}.

The remainder of the article is organized as follows. 
In Section~\ref{sec:prb}, we discuss our modeling framework and formulate an MIQP for the considered problem.
In Section~\ref{sec:alg}, we present a distributed optimization algorithm to solve the MIQP problem and its theoretical results.
We show the simulation results in Section~\ref{sec:sim} and conclude the paper in Section~\ref{sec:conc}.

\section{Problem Formulation}
\label{sec:prb}

In this section, we present our modeling framework and formulate a joint optimization problem for traffic light control and CAV coordination.

\subsection{System Model}

We consider an isolated intersection illustrated in Fig.~\ref{fig:scenario}, which has separate lanes for right turns, straight-through traffic, and left turns.
We define a control zone in the vicinity of the intersection, where the vehicles and the traffic light controllers (TLCs) have communication capability and can be managed by the proposed framework.
We assume that lane changing is not allowed in the control zone.
Note that though the scenario in Fig.~\ref{fig:scenario} consists of 12 lanes, the right-turn lanes can be considered separately since they do not have any potential conflicts with other lanes, which reduces the problem to 8 lanes of interest. 
Next, we introduce several definitions for the entities involved in the intersection scenario.

\begin{definition}
(Lanes) Let \lane{l} be the $l$-th lane in the scenario and $\LLL$ be the set of all lanes' indices.
Each \lane{l} is equipped with a traffic light controller denoted by \TLC{l}.
We set each lane's origin location at the control zone's entry and let $\psi_l$ be the position of the stop line along \lane{l}.
\end{definition}

\begin{definition}
(Nodes) Let \node{n} and the notation $n=l\otimes m$ to denote that \node{n} is the conflict point between \lane{l} and \lane{m}. 
Let $\phi_l^n$ and $\phi_m^n$ be the positions of \node{n} along \lane{l} and \lane{m}, respectively.
\end{definition}

\begin{definition}
(Vehicles) Let a tuple $(i,l)$ be the index of the $i$-th vehicle traveling in \lane{l}.
For each \lane{l}, let $\CCC_l(k)$ and $\HHH_l(k)$ be set of CAVs and HDVs in \lane{l} at any time step $k$.
\end{definition}

\begin{figure}
\centering
\includegraphics[width=0.4\textwidth]{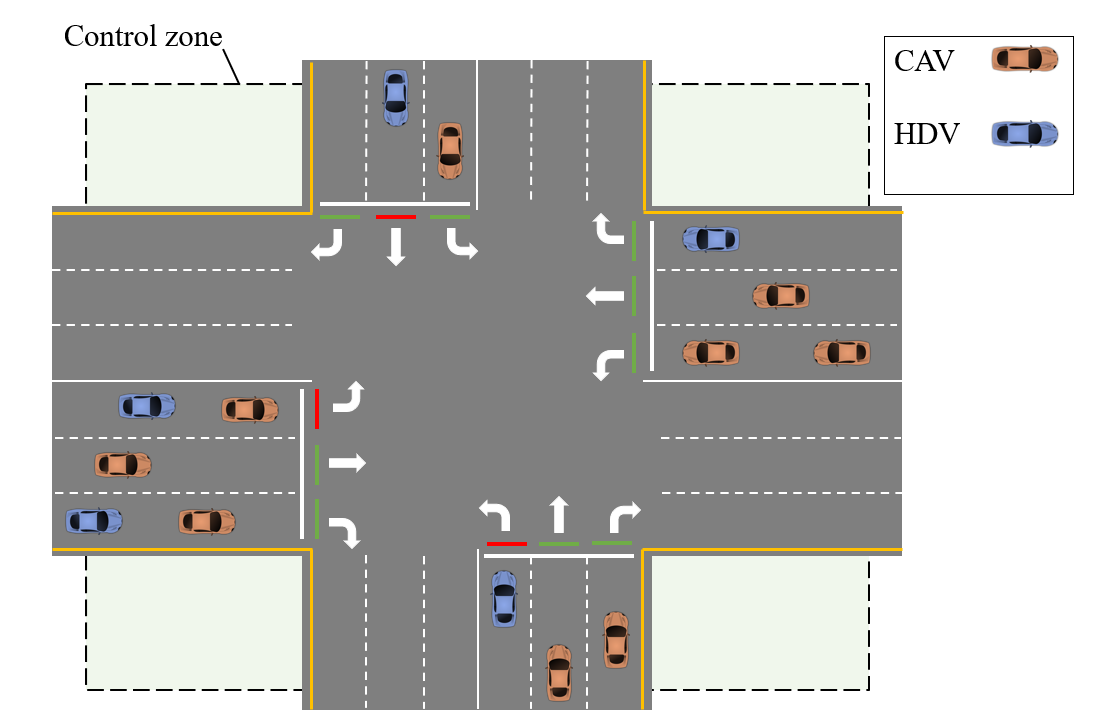}
\caption{An intersection scenario in mixed traffic with 12 lanes, including separate lanes for right turns, straight-through traffic, and left turns.}
\label{fig:scenario}
\end{figure}

\hilitediff{We formulate a joint optimization problem for TLCs and CAVs and implement it in a receding horizon framework to ensure robustness against the uncertainty caused by HDVs.}
Let $k_0$ be the current time step, $H$ be the control horizon, $\III = \{k_0+1, \dots, k_0+H\}$ be the set of time steps in the next control horizon, and $\Delta T$ be the sample time. 
In our optimization problem, we consider at most one traffic light switch for each lane over the control horizon.
\hilitediff{Let $\kappa_{l} \in \{1, \dots, H+1\}$ so that the traffic light switches at $k_0+\kappa_{l}$.
If $\kappa_l = 1$, the traffic light switches at the next time step, while if $\kappa_l = H$, the traffic light switches at the last time step of the control horizon, and $\kappa_l = H+1$ means that the traffic light does not switch over the next control horizon.} 
At each time step $k$, let $s_l(k) \in \{0, 1\}$ be the traffic light state of \lane{l}, $l \in \LLL$, where $s_l(k) = 0$ if the traffic light is red and $s_l(k) = 1$ if the traffic light is green.
The traffic light states can be modeled as follows,
\begin{equation}
\label{eq:tlmodel}
s_l (k) =
\begin{cases}
1-s_l (k_0), & \text{if } k \ge k_0 + \kappa_{l}, \\
s_l (k_0), & \text{otherwise,}
\end{cases}     
\end{equation}
for all $k \in \III$. 
In addition, we impose a constraint for the minimum and maximum time gaps between traffic light switches.
Let $\bar{k}_l$ be the time of the last traffic light switch of \lane{l}, $\delta_{\min} \in \RRplus$ and $\delta_{\max} \in \RRplus$ be the minimum and maximum time gaps, respectively. 
As a result, the next switching time must be within $[\delta_{\min}+\bar{k}_l, \delta_{\max}+\bar{k}_l]$.
Since $\kappa_l \in \{1, \dots, H+1\}$, to guarantee feasibility for $\kappa_l$, the constraints are formulated as follows, 
\begin{equation}
\label{eq:min-swt}
\begin{split}
\kappa_{l} &\ge \min \{ \delta_{\min}+\bar{k}_l-k_0, H+1\}, \\
\kappa_{l} &\le \max \{ \delta_{\max}+\bar{k}_l-k_0, 1\}.
\end{split}
\end{equation}
Note that the constraints \eqref{eq:min-swt} can be ignored if all vehicles in \lane{l} are CAVs.

\hilitediff{For each \CAV{l,i}, let $p_{l,i}(k) \in \RR$, $v_{l,i}(k) \in \RR$, and $u_{l,i}(k) \in \RR$ be the position, speed, and control input (acceleration/deceleration) at time step $k$.}
We consider the discrete double-integrator dynamics for \CAV{l,i} as follows,
\begin{equation}
\label{eq:integrator}
\hilitediff{
\begin{split}
p_{l,i}(k) &= p_{l,i}(k\!-\!1) \!+\! \Delta T v_{l,i}(k\!-\!1) \!+\! \frac{1}{2} \Delta T^2 u_{l,i}(k\!-\!1),  \\
v_{l,i}(k) &= v_{l,i}(k\!-\!1) \!+\! \Delta T u_{l,i}(k\!-\!1),
\end{split}}
\end{equation}
for all $k \in \III$. In addition, we impose the following speed and acceleration limit constraints,
\begin{equation}
\label{eq:bound-cons}
v_{\min} \le v_{l,i} (k) \le  v_{\max}, \; u_{\min} \le u_{l,i} (k\!-\!1) \le  u_{\max},
\end{equation}
for all $k \in \III$, where ${u_{\min}\in \RRminus}$ and ${u_{\max}\in \RRplus}$ are the minimum and maximum control inputs, ${v_{\min}\in \RRplus}$ and ${v_{\max}\in \RRplus}$ are the minimum and maximum speeds, respectively.

\hilitediff{In this work, we consider that the trajectories of HDVs over the control horizon are predicted using the constant acceleration model, that is, each HDV’s predicted acceleration rate remains constant over the control horizon while satisfying the speed constraint in \eqref{eq:bound-cons}.
The constant acceleration rate of each HDV at each time step is approximated by averaging past data.}
We assume that the predicted trajectories for each HDV, based on the constant acceleration model, are computed by the TLC of the lane and then transmitted to the CAVs.

\begin{remark}
Although the constant acceleration model may not precisely predict the actual behavior of HDVs, the receding horizon implementation effectively manages this discrepancy. Developing a more accurate prediction model for HDVs is beyond the scope of this paper.
\end{remark}

\subsection{Coupling constraints}

Coordinating traffic lights and CAVs requires them to satisfy several coupling constraints.
First, the traffic lights must guarantee no lateral conflicts between a CAV and an HDV or between two HDVs.
Let \lane{l} and \lane{m} be two lanes with lateral conflicts, then the lights for those lanes cannot be both green if there are HDV-CAV or HDV-HDV conflicts, \ie
\begin{equation}
\label{eq:no-conf}
s_l(k) + s_m(k) \le 1,\; \forall k \in \III, \;\text{if}\; \eta(l, m, k_0) > 0,
\end{equation}
where $\eta (l,m,k_0)$ counts the number of CAV-HDV or HDV-HDV pairs that have lateral conflicts between \lane{l} and \lane{m} at the current time step $k_0$.
If a pair of vehicles travel on two intersecting lanes and neither vehicle has yet crossed the conflict point, they have lateral conflicts.
Next, we impose a constraint ensuring that the CAVs stop at red lights,
\begin{equation}
\label{eq:red-stop}
p_{l,i} (k) \le \psi_l,\; \forall k \in \III, \; \text{if} \; s_l(k) = 0.
\end{equation}
\hilitediff{Note that we only need to impose \eqref{eq:red-stop} if \CAV{l,i} is the first vehicle in the queue of \lane{l}, \ie \CAV{l,i} has not crossed the stop line, and can stop by the line under maximum deceleration.}
To avoid rear-end collisions, we consider safety constraints for each \CAV{l,i} if there is a preceding vehicle, which can be either a CAV or an HDV, as follows
\begin{equation}
\label{eq:rearend}
p_{l,i} (k) + \tau v_{l,i} (k) + d_{\rm{min}} \le p_{l,i-1} (k),\; \forall k \in \III,
\end{equation}
where $\tau \in \RRplus$ is the desired time headway, and $d_{\rm{min}}\in \RRplus$ is the minimum distance.
Note that if \veh{l,i-1} is an HDV, then \eqref{eq:rearend} is considered as a local constraint of \CAV{l,i}, while \veh{l,i-1} is a CAV, \eqref{eq:rearend} is a coupling constraint between the two CAVs.
We also consider lateral safety constraints between two CAVs, \eg \CAV{l,i} and \CAV{m,j}, traveling on two lanes with lateral conflict node $n = l \otimes m$.
The lateral safety constraints can be formulated using OR statements \cite{alrifaee2014centralized}. %
The OR statement can be equivalently formulated using the following set of linear constraints,
\begin{equation}
\label{eq:lateral-mi}
\begin{split}
p_{l,i} (k) - \phi^n_{l} & \ge d_{\rm{min}} - M (1 - c_{l,i}), \\
\phi^n_{l} - p_{l,i} (k) & \ge d_{\rm{min}} - M (1 - e_{l,i}), \\
p_{m,j} (k) -\phi^n_{m} & \ge d_{\rm{min}} - M (1 - c_{m,j}), \\
\phi^n_{m} - p_{m} (k) & \ge d_{\rm{min}} - M (1 - e_{m,j}), 
\end{split}
\end{equation}
and 
\begin{equation}
\label{eq:lateral-1}
c_{l,i} + e_{l,i} + c_{m,j} + e_{m,j} \ge 1,
\end{equation}
for all $k \in \III$, where $c_{l,i}, e_{l,i} \in \{0, 1\}$, and $M \in \RRplus$ is a large positive number. 
\hilitediff{The constraint \eqref{eq:lateral-1} implies that at least a binary among $c_{l,i}$, $e_{l,i}$, $c_{m,j}$, and $e_{m,j}$ must be $1$, which, when combined with \eqref{eq:lateral-mi}, guarantees the distance of at least one vehicle to the conflict point is greater than or equal to $d_{\rm{min}}$.}
Note that all the binary variables $c_{l,i}$ and $e_{l,i}$, $\forall i \in \CCC_l(k_0)$  are handled by \TLC{l}. %

\subsection{Objective Function}

In our formulation, each agent has its separate objective.
For the TLCs, the goal is to maximize the traffic throughput by encouraging the traffic lights for the lanes with higher priority to be green. 
Let $\gamma_{l}$ be the priority coefficient of \lane{l}, which depends on the number of vehicles upstream of the stop line.
Moreover, to put a higher priority on the vehicles near the intersection, we use the sigmoid function that takes the vehicles' current positions as inputs, leading to the following computation at the current time step $k_0$,
\begin{equation}
\gamma_{l} = \sum_{ \substack{i \in \CCC_l(k_0) \cup \HHH_l(k_0), \\ p_{l,i} (k_0) < \psi_l} } \mathrm{sigmoid} \bigg( \frac{p_{l,i} (k_0) - \psi_l/2}{\psi_l/2} \bigg).
\end{equation}
For the CAVs, we optimize the trajectories over the control horizon by a weighted sum of two terms: (1) maximization of the positions to improve the travel time and (2) minimization of the acceleration rates for energy savings.
Thus, our optimization problem minimizes the following objective,
\hilitediff{
\begin{equation}
\label{eq:obj}
\begin{multlined}
\underset{\substack{\bbsym{s}_{l}, \bbsym{u}_{l,i} \\ \forall l \in \LLL, i \in \CCC_l(k_0)}}{\minimize} \quad
\sum_{k \in \III} \sum_{l \in \LLL} \Bigg( - \gamma_{l} \, s_l(k+1) \\ 
+ \sum_{i \in \CCC_l(k_0)} - \omega_p p_{l,i} (k+1)  + \omega_u u^2_{l,i} (k) \Bigg),
\end{multlined}
\end{equation}
where $\omega_p \in \RRplus$ and $\omega_u \in \RRplus$ are the positive weights, while we let $\bbsym{s}_{l} = [s_{l}(k)]_{k \in \III}$, $\bbsym{u}_{l,i} = [u_{l,i}(k)]_{k \in \III}$ be the vectors of decision variables for \TLC{l} and \CAV{l,i} over the control horizon, respectively.
}

In our problem formulation, the objective function is quadratic, while all the constraints are linear, which results in an MIQP problem.

\section{Distributed Mixed-Integer Optimization Algorithm}
\label{sec:alg}

In this section, we present the distributed algorithm to solve the MIQP problem at every time step inspired by the MBI method \cite{li2015convergence}.
It is worth noting that there is a limited amount of research on distributed algorithms for MIQPs in the existing literature, \eg \cite{sun2018distributed,liu2021distributed,klostermeier2024numerical}.

\subsection{Distributed Mixed-Integer Optimization Problem}

To facilitate the exposition of the algorithm, we rewrite the optimization problem in the previous section in a compact form.
First, we describe the communication between the agents in our multi-agent optimization problem as follows. 
\begin{definition}
Let $N$ be the total number of agents, $\VVV = \{1, \dots, N\}$ be the set of agents, and $\EEE \subset \VVV \times \VVV$ be the set of communication links.
For each \agent{i}, $i \in \VVV$, let $\NNN_i = \{j \in \VVV \;|\; (i,j) \in \EEE\}$ be the set of its neighbors.
\end{definition}

\hilitediff{Note that for our problem, the agents include TLCs and CAVs, and a pair of agents must have a communication link if they share coupling constraints.}
The problem at every time step can be formulated as a mixed-integer optimization problem with separable objectives and pair-wise coupling constraints as follows,
\begin{subequations}
\label{eq:compact-prb}
\begin{align}
\underset{\bbsym{x}_i \in \XXX_i, \forall i \in \VVV}{\minimize} & \quad \sum_{i \in \VVV} f_i (\bbsym{x}_i), \label{eq:compact-prb-obj} \\
\text{subject to} & \quad  \bbsym{A}_{i} \bbsym{x}_i \le \bbsym{b}_{i}, \, i \in \VVV, \label{eq:compact-prb-loc} \\
& \quad (\bbsym{C}^{j}_{i})^\top \bbsym{x}_i + (\bbsym{C}^{i}_{j})^\top \bbsym{x}_{j} \le \bbsym{d}_{ij}, \, (i,j) \in \EEE, \label{eq:compact-prb-cpl}
\end{align}
\end{subequations}
where 
$\XXX_i$ is the domain set for $\bbsym{x}_i$,
$f_i (\bbsym{x}_i) = \bbsym{x}_i^\top \bbsym{Q}_i \bbsym{x}_i + \bbsym{q}_i^\top \bbsym{x}_i$ is the local objective function, $\bbsym{Q}_i$, $\bbsym{q}_i$, $\bbsym{A}_{i}$, $\bbsym{b}_{i}$, $\bbsym{C}^{j}_{i}$, $\bbsym{C}^{i}_{j}$, and $\bbsym{d}_{ij}$ are the matrices and vectors of coefficients.
The constraint \eqref{eq:compact-prb-loc} collects all the (linear) local constraints of \agent{i}, while the constraint \eqref{eq:compact-prb-cpl} collects all the (linear) coupling constraints between \agent{i} and \agent{j}. 
Next, we impose the following assumption for the feasibility of the optimization problem \eqref{eq:compact-prb}.
\begin{assumption}
\label{assp:feasibility}
The problem \eqref{eq:compact-prb} admits at least a feasible solution.
\end{assumption}

\subsection{Penalization-Enhanced Maximum Block Improvement Algorithm}

In the MBI algorithm, each agent optimizes its local variables at each iteration while keeping all other agents' variables fixed.
Only the agent that achieves the greatest cost reduction among its neighbors can update its local variables.
However, the MBI algorithm requires an initial feasible solution, which is not always trivial for our problem. %
To overcome this, we propose to use the maximum penalty function to relax the local inequality constraints so that %
finding an initial feasible solution that satisfies the coupling constraints becomes trivial.
For example, if all the local constraints are relaxed, a trivial solution that satisfies all the coupling constraints in our problem can be found by setting $s_{l} (k+1) = 0$ and $v_{l,i} (k+1) = 0$, $\forall k \in \III$, \ie all the traffic lights are red, and all the CAVs stop over the control horizon.
We prove that if the penalty weights are chosen appropriately, the penalized problem can yield the same solution as the original problem.
Let $\rho \in \RRplus$ be the penalty weight for the local constraint relaxation.
The penalized version of \eqref{eq:compact-prb} is given by
\begin{subequations}
\label{eq:pen_prb}
\begin{align}
\underset{\bbsym{x}_i \in \XXX_i, \forall i \in \VVV}{\minimize} & \quad \sum_{i \in \VVV} \tilde{f}_i (\bbsym{x}_i), \\
\text{subject to} & \quad (\bbsym{C}^{j}_{i})^\top \bbsym{x}_i + (\bbsym{C}^{i}_{j})^\top \bbsym{x}_{j} \le \bbsym{d}_{ij}, \, (i,j) \in \EEE,
\end{align}
\end{subequations}
where 
\begin{equation}
\tilde{f}_i (\bbsym{x}_i) = f_i (\bbsym{x}_i) + \rho. \bbsym{1}^\top \max(\bbsym{0}, \bbsym{A}_{i} \bbsym{x}_i - \bbsym{b}_{i}).
\end{equation}

The MBI algorithm for solving \eqref{eq:pen_prb} consists of the following steps.
The algorithm starts with an initial feasible solution $\{\bbsym{x}_i^{(0)}\}_{i \in \VVV}$.
Then, at each iteration $t$, each \agent{i} transmits $(\bbsym{C}^{i}_{j})^\top \bbsym{x}_i^{(t)}$ to \agent{j}, $\forall j \in \NNN_i$. 
Next, in parallel, \agent{i} solves the following local sub-problem given fixing the other agents' variables to obtain a new candidate solution $\hat{\bbsym{x}}_i^{(t+1)}$,
\begin{subequations}
\label{eq:loc-subprb}
\begin{align}
\hat{\bbsym{x}}_i^{(t+1)} =
\underset{\bbsym{x}_i \in \XXX_i}{\argmin} & \quad \tilde{f}_i (\bbsym{x}_i), \\
\text{subject to} & \quad (\bbsym{C}^{j}_{i})^\top \bbsym{x}_i \le \bbsym{d}_{ij} - (\bbsym{C}^{i}_{j})^\top \bbsym{x}_{j}^{(t)},
\end{align}
\end{subequations}
and compute the local cost reduction $\Delta f_i^{(t+1)} = \tilde{f}_i(\hat{\bbsym{x}}_i^{(t+1)}) - \tilde{f}_i(\bbsym{x}_i^{(t)})$ for the new candidate solution.
Then \agent{i} transmits $\Delta f_i^{(t+1)}$ to \agent{j}, $\forall j \in \NNN_i$.
\Agent{i} accepts or rejects the new solution by comparing its local cost reduction with its neighbors, \ie $\bbsym{x}_i^{(t+1)} = \hat{\bbsym{x}}_i^{(t+1)}$ if $\Delta f_i^{(t+1)} > \Delta f_j^{(t+1)}, \; \forall j \in \NNN_i$, otherwise $\bbsym{x}_i^{(t+1)} = \bbsym{x}_i^{(t)}$.
The algorithm terminates if no agent can further improve its local variables.
As the TLCs handle all the integer variables, the sub-problems for CAVs are quadratic programs (QPs), while those for TLCs are linear integer programs (LIPs).
In what follows, we analyze the theoretical properties of the proposed algorithm.

\begin{lemma}
\label{lem:1}
\hilitediff{If the non-penalized sub-problem \eqref{eq:pen_prb} for each CAV is feasible and the penalty weight $\rho$ is chosen sufficiently large, the solution of the penalized sub-problem is also the solution of the non-penalized sub-problem.}
\end{lemma}

\begin{proof}
For CAVs, the sub-problems are QPs, thus this lemma is a special case of Proposition 5.25 in \cite{bertsekas2014constrained}. 
\end{proof}

\begin{lemma}
\label{lem:2}
\hilitediff{If the non-penalized sub-problem \eqref{eq:pen_prb} for each TLC is feasible and the penalty weight $\rho$ is chosen sufficiently large, the solution of the penalized sub-problem is also the solution of the non-penalized sub-problem.}
\end{lemma}
  
\begin{proof}
We prove this lemma based on the fact that the domain set $\XXX_i$ for the integer variables is finite.
Let $\tilde{\bbsym{x}}_i^*$ be the minimizer of the penalized sub-problem. 
If $\tilde{\bbsym{x}}_i^*$ satisfies all the local constraints, then obviously $\tilde{\bbsym{x}}_i^*$ is the feasible minimizer of the non-penalized sub-problem.
Therefore, we only need to consider if $\tilde{\bbsym{x}}_i^*$ violates at least a local constraint.
Let $\bbsym{x}_i^*$ be the minimizer of the non-penalized sub-problem. 
From the definition of $\tilde{\bbsym{x}}_i^*$, we have
\begin{equation}
\label{eq:lem2-1}
f_i (\tilde{\bbsym{x}}_i^*) + \rho. \bbsym{1}^\top \max(0, \bbsym{A}_{i} \tilde{\bbsym{x}}_i^* - \bbsym{b}_{i}) \le f_i (\bbsym{x}_i^*),
\end{equation}
since $\max(\bbsym{0}, \bbsym{A}_{i} \bbsym{x}_i^* - \bbsym{b}_{i}) = \bbsym{0}$.
As $\tilde{\bbsym{x}}_i^*$ violates at least a local constraint, there exists $\epsilon > 0$ such that $\bbsym{1}^\top \max(\bbsym{0}, \bbsym{A}_{i} \tilde{\bbsym{x}}_i^* - \bbsym{b}_{i}) > \epsilon$.
Combining with \eqref{eq:lem2-1}, we obtain
\begin{equation}
\label{eq:lem2-2}
\rho \epsilon \le  f_i (\bbsym{x}_i^*) - f_i (\tilde{\bbsym{x}}_i^*).
\end{equation}
Since $\XXX_i$ is a finite set and $f_i$ is a linear function, $f_i$ is bounded over $\XXX_i$.
In other words, there exist $0 < \xi < +\infty$ such that $f_i (\bbsym{x}_i^*) - f_i (\tilde{\bbsym{x}}_i^*) \le \xi$. 
If we choose $\rho > \xi/\epsilon$, it contradicts with \eqref{eq:lem2-2}.
Therefore, if $\rho$ is sufficiently large, $\tilde{\bbsym{x}}_i^*$ must be local-constraint feasible for the non-penalized sub-problem. 
\end{proof}

\hilitediff{In the next definition, we extend the concept of person-by-person optimality \cite{bauso2008generalized} to constrained multi-agent optimization.
\begin{definition}
\label{def:pbp}
(Person-by-person optimality) $\{\bbsym{x}_i^*\}_{i \in \VVV}$ is a person-by-person optimal solution of \eqref{eq:compact-prb} if for any $i \in \VVV$,
\begin{equation}
\begin{split}    
\bbsym{x}_i^* = \;\underset{\bbsym{x}_i \in \XXX_i}{\argmin} & \; f_i (\bbsym{x}_i), \\
\text{subject to} & \;  \bbsym{A}_{i} \bbsym{x}_i \le \bbsym{b}_{i}, \\
& \; (\bbsym{C}^{j}_{i})^\top \bbsym{x}_i + (\bbsym{C}^{i}_{j})^\top \bbsym{x}_{j}^* \le \bbsym{d}_{ij}, \, \forall j \in \NNN_i.
\end{split}
\end{equation}
and we refer to $\bbsym{x}_i^*$ as the best feasible response to $\{\bbsym{x}_{j}^*\}_{j \in \NNN_i}$.
\end{definition}}

\hilitediff{Note that the person-by-person optimality condition for the penalized problem \eqref{eq:pen_prb} follows similarly to Definition~\ref{def:pbp} and is therefore omitted.}

\begin{lemma}
\label{lem:3}
Starting from an initial solution $\{\bbsym{x}_i^{(0)}\}_{i \in \VVV}$ that satisfies the coupling constraints, the sequence $\{\bbsym{x}_i^{(t)}\}_{i \in \VVV}$ generated by the algorithm at any iteration satisfies the coupling constraints.
In addition, if the generated sequence has an \hilitediff{accumulation point}, it is a person-by-person optimal solution of the penalized problem \eqref{eq:pen_prb}.
\end{lemma}

\begin{proof}
First, we can observe that the feasibility of each local sub-problem \eqref{eq:loc-subprb} is always maintained since $\bbsym{x}_i^{(t)}$ itself is feasible. 
We prove that if $\{\bbsym{x}_i^{(t)}\}_{i \in \VVV}$ is coupling-constraint feasible for \eqref{eq:compact-prb}, so is $\{\bbsym{x}_i^{(t+1)}\}_{i \in \VVV}$ .
Note that if \agent{i} and \agent{j} have coupling constraints, then at each iteration, at most, one of them can accept its solution.
That means %
$(\bbsym{x}_i^{(t+1)}, \bbsym{x}_j^{(t+1)})$ can be either $(\bbsym{x}_i^{(t)}, \bbsym{x}_j^{(t)})$, $(\hat{\bbsym{x}}_i^{(t+1)}, \bbsym{x}_j^{(t)})$, or $(\bbsym{x}_i^{(t)}, \hat{\bbsym{x}}_j^{(t+1)})$, which all satisfy the coupling constraints. 
Therefore, all the coupling constraints are satisfied by $\{\bbsym{x}_i^{(t+1)}\}_{i \in \VVV}$.
Let $\{\bar{\bbsym{x}}_i\}_{i \in \VVV}$ be a \hilitediff{accumulation point} of the sequence $\{\bbsym{x}_i^{(t)}\}_{i \in \VVV}$ in the sense that for each $i \in \VVV$ there exists a subsequence $\{\bbsym{x}_i^{(t_s)}\}$ that converges to $\bar{\bbsym{x}}_i$. 
Note that since $\XXX_i$ is closed, $\forall i \in \VVV$, $\bar{\bbsym{x}}_i \in \XXX_i$.
Moreover, as the coupling constraints generate a closed polyhedron and the sequence $\{\bbsym{x}_i^{(t)}\}_{i \in \VVV}$ is coupling-constraint feasible for all $t$, $\{\bar{\bbsym{x}}_i\}_{i \in \VVV}$ is also coupling-constraint feasible.

For any $i \in \VVV$, let $\bar{\bbsym{x}}_i^*$ be the best feasible response to $\{\bar{\bbsym{x}}_{j}\}_{j \in \NNN_i}$, \ie
\begin{equation}
\label{eq:lem3-1}
\bar{\bbsym{x}}_i^* = \underset{\bbsym{x}_i \in \XXX_i}{\argmin} \quad \tilde{f}_i (\bbsym{x}_i) + \II_{cc} (\bbsym{x}_i, \{\bar{\bbsym{x}}_{j}\}_{j \in \NNN_i}),
\end{equation}
where to simplify the proof, we utilize $\II_{cc} (\bbsym{x}_i, \{\bar{\bbsym{x}}_{j}\}_{j \in \NNN_i})$ as the indicator function that take $+\infty$ if the coupling constraint is violated.
We have
\begin{equation}
\label{eq:lem3-2}
\tilde{f}_i (\bar{\bbsym{x}}_i^*) + \II_{cc} (\bar{\bbsym{x}}_i^*, \{\bbsym{x}_j^{(t_s)}\}_{j \in \NNN_i})
\ge \tilde{f}_i (\bbsym{x}_i^{(t_s+1)}),
\end{equation}
since $\bbsym{x}_i^{(t_s+1)}$ is the best feasible response to $\{\bbsym{x}_j^{(t_s)}\}_{j \in \NNN_i}$.
Moreover, as $\tilde{f}_i (\bbsym{x}_i^{t})$ is a non-increasing sequence, from \eqref{eq:lem3-2} we obtain
\begin{equation}
\tilde{f}_i (\bar{\bbsym{x}}_i^*) + \II_{cc} (\bar{\bbsym{x}}_i^*, \{\bbsym{x}_j^{(t_s)}\}_{j \in \NNN_i}) \ge \tilde{f}_i (\bbsym{x}_i^{(t_{s+1})}).
\end{equation}
If $t \rightarrow \infty$ it yields
\begin{equation}
\label{eq:lem3-3}
\tilde{f}_i (\bar{\bbsym{x}}_i^*)  = \tilde{f}_i (\bar{\bbsym{x}}_i^*) + \II_{cc} (\bar{\bbsym{x}}_i^*, \{\bar{\bbsym{x}}_j^*\}_{j \in \NNN_i}) \ge \tilde{f}_i (\bar{\bbsym{x}}_i).
\end{equation}
Since $\bar{\bbsym{x}}_i$ is also a feasible response to $\{\bar{\bbsym{x}}_{j}\}_{j \in \NNN_i}$, \eqref{eq:lem3-1} and \eqref{eq:lem3-3} implies that $\tilde{f}_i (\bar{\bbsym{x}}_i^*) = \tilde{f}_i (\bar{\bbsym{x}}_i)$, or $\bar{\bbsym{x}}_i$ is the best feasible response to $\{\bar{\bbsym{x}}_{j}\}_{j \in \NNN_i}$.
\end{proof}

Combining Lemmas~\ref{lem:1}, \ref{lem:2} and \ref{lem:3}, we obtain the following main theorem. 
\begin{theorem}
\hilitediff{The accumulation point} generated from the penalization-enhanced MBI algorithm is a person-by-person optimal solution of \eqref{eq:compact-prb}.
\end{theorem}

\section{Simulation Results}
\label{sec:sim}

In this section, we demonstrate the control performance of the proposed framework by numerical simulations.

\subsection{Simulation Setup}

\hilitediff{We validated our framework using a mixed-traffic simulation environment in \texttt{SUMO} interfacing with Julia programming language via \texttt{TraCI} \cite{lopez2018microscopic} and \texttt{PyCall} package.}
In the simulation, we considered an intersection with a control zone of length \SI{250}{m} for each lane, with \SI{150}{m} upstream and \SI{100}{m} downstream of the stop line position.
We conducted multiple simulations for three traffic volumes: $1200$, $1400$, and $1600$ vehicles per hour and six different penetration rates: $0\%$, $20\%$, $40\%$, $60\%$, $80\%$, and $100\%$.
The vehicles are randomly assigned to the lanes, with higher entry rates for left turns and straight-through traffic, while the entry rates across four directions are balanced.
In our implementation, \texttt{GUROBI} optimizer \cite{gurobi} is used to solve the QP and LIP sub-problems.

\subsection{Results and Discussions}

Videos and data of the simulations can be found at \url{https://sites.google.com/cornell.edu/tlc-cav}.
The simulations demonstrate that the proposed framework effectively coordinates TLCs and CAVs across varying traffic volumes and penetration rates.
The framework allows CAVs to coordinate their intersection crossings while mitigating full stops at traffic lights, especially at high penetration rates. 
To assess the framework's performance extensively, we computed the average travel time and average acceleration rates from the simulation data, with a duration of $1800$ seconds for each simulation.
The results are shown in Fig.~\ref{fig:metr}.
Overall, in all three examined traffic volumes, starting from a penetration rate of $60\%$, the framework achieves remarkable travel time improvement compared to the scenario with pure HDVs.
The framework also performs well at lower penetration rates ($20\%$ and $40\%$), and differences compared to the pure HDV case remain relatively small.
Moreover, we can observe that CAV penetration generally leads to lower average acceleration rates, which may result in better energy consumption and a more comfortable travel experience. 

To demonstrate the benefits of leveraging CAV coordination for intersection crossing, we compare the average travel time between using the proposed control approach ($\#1$) and an alternative approach ($\#2$) where the lateral conflicts between CAVs are handled by the traffic lights instead of by the lateral constraints \eqref{eq:lateral-mi} and \eqref{eq:lateral-1}.
The statistics are computed from simulations conducted with the high traffic volume of $1600$ vehicles per hour and are shown in Fig.~\ref{fig:baseline} .  
As can be seen from the figure, utilizing CAV coordination, approach $\#1$ outperforms in terms of travel time the approach with only intelligent traffic light control across all penetration rates. 
The improvement is especially significant at lower penetration rates, highlighting the advantages of combining CAV coordination with traffic light control for conflict resolution compared to relying solely on traffic lights.
Approach $\#1$ may require higher acceleration rates than approach $\#2$, but it can be explained by the fact that the vehicles experience more full stops in simulations using approach $\#2$.

\hilitediff{Finally, we compare the average computation time between our proposed distributed algorithm and the centralized optimization approach employing the \texttt{GUROBI} solver to solve the MIQP problem. 
Table~\ref{tab:time} presents the average computation time across different numbers of agents for the two approaches.
Overall, our distributed algorithm required less time than the centralized optimization approach, validating the benefits of distributed computation.
However, the results reveal that the computation time for the distributed algorithm still increases with the number of agents. 
This is because as the number of agents grows, the algorithm may require more iterations to achieve convergence.}

\begin{figure}[!tb]
\centering
\includegraphics[scale=0.28]{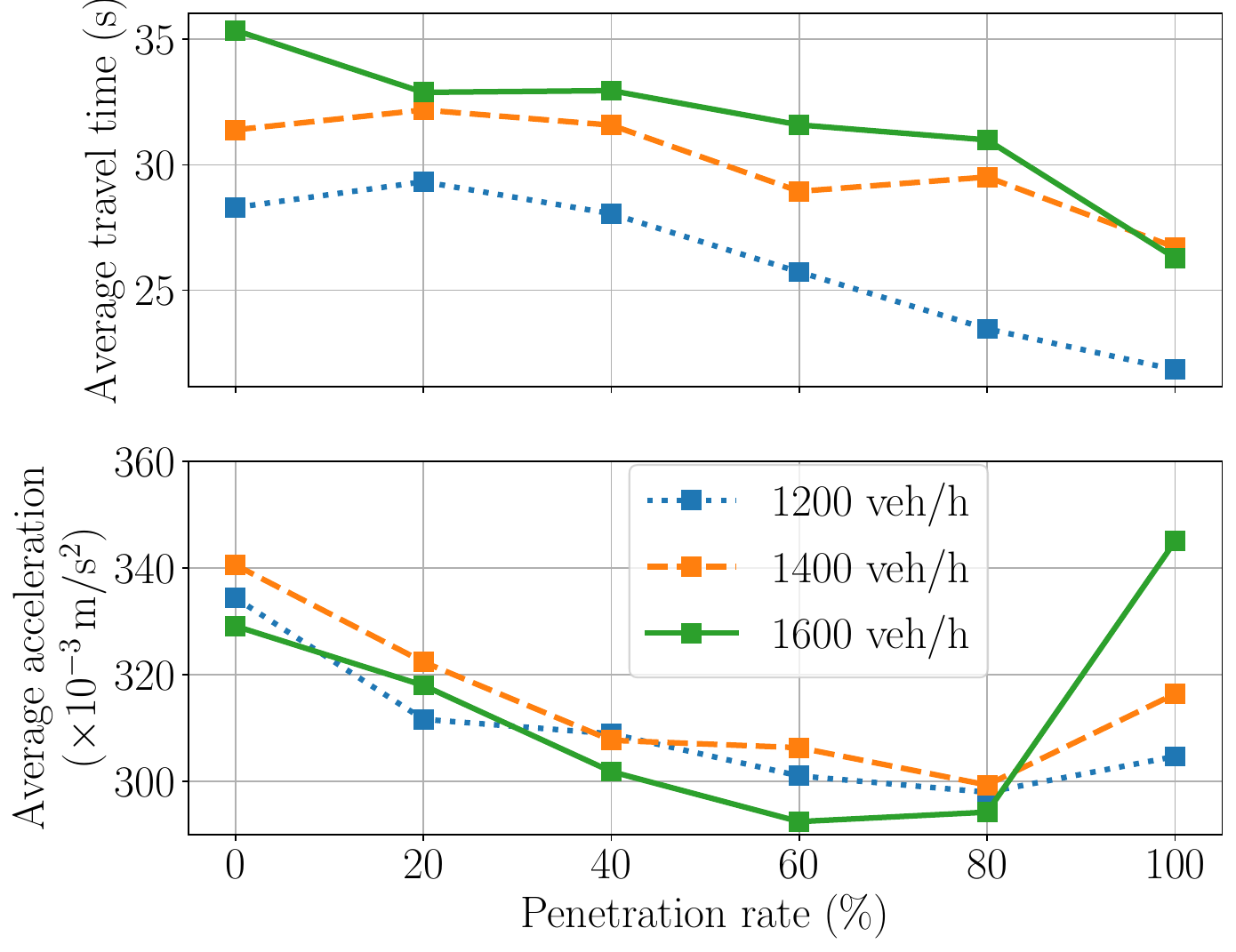}
\caption{Average travel time (top) and average acceleration (bottom) under different penetration rates and traffic volumes.}
\label{fig:metr}
\end{figure}
\begin{figure}
\centering
\includegraphics[scale=0.28]{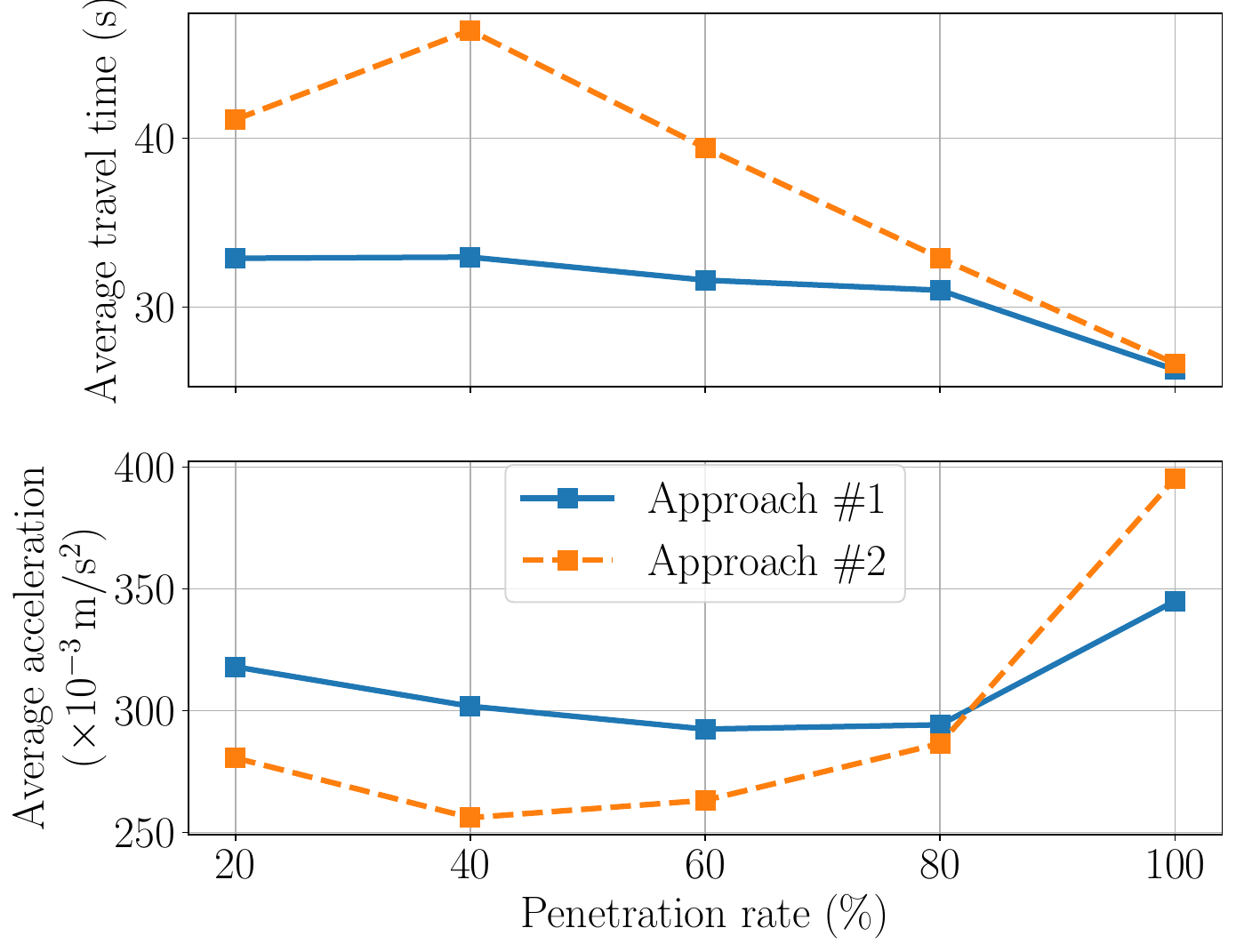}
\caption{Comparison of average travel time (top) and average acceleration (bottom) between two approaches with ($\#1$) and without ($\#2$) CAV coordination for lateral conflicts.}
\label{fig:baseline}
\end{figure}

\begin{table}[!t]

\centering
\caption{\hilitediff{Average computation time (in mili-seconds) for the centralized optimization using \texttt{GUROBI} solver and the proposed distributed optimization algorithm across different number of agents.}}
\label{tab:time}
\hilitediff{
\begin{tabular}{ K{1.2cm} | K{1.0cm} | K{1.0cm} | K{1.0cm} | K{1.0cm} | K{1.0cm}  }
\toprule[1.0pt]
& {$N = 10$} & {$N = 15$} & {$N = 20$} & {$N = 25$} & {$N = 30$} \\ 
\midrule[0.5pt]
Distributed
& $\bf{12.3}$ 
& $\bf{16.6}$ 
& $\bf{23.3}$ 
& $\bf{33.6}$ 
& $\bf{53.2}$ \\
Centralized
& ${13.3}$ 
& ${36.5}$ 
& ${76.8}$ 
& ${126.7}$  
& ${167.1}$ \\
\bottomrule[1.0pt]
\end{tabular}
}
\end{table}

\section{Conclusions}
\label{sec:conc}

In this paper, we addressed the optimal coordination of traffic lights and CAVs in mixed-traffic intersections by formulating an MIQP problem.
The formulation improves upon the previous ones by allowing the traffic lights of two lanes with lateral conflicts to be green if no conflict involving an HDV, enabling better CAV coordination at higher penetration rates.
We also propose the penalization-enhanced MBI algorithm to find a feasible person-by-person optimal solution in a distributed manner.
Future work should focus on enhancing HDV trajectory prediction and developing distributed algorithms that do not require an initial feasible solution.

\bibliographystyle{IEEEtran}
\bibliography{IEEEabrv,bib/refs_opt,bib/refs_cav,bib/refs_mt,bib/refs_IDS,bib/refs_others}

\balance

\end{document}